\documentclass[%
aip,
jmp,
 amsmath,
 amssymb,
 reprint,%
]{revtex4-1}

\usepackage{graphicx}
\usepackage{dcolumn}
\usepackage{bm}

\usepackage[utf8]{inputenc}
\usepackage[T1]{fontenc}
\usepackage{etoolbox}
\usepackage{ascmac,wrapfig,makeidx}
\usepackage{physics}
\usepackage[stable]{footmisc}
\usepackage{mathrsfs}
\usepackage{color}
\usepackage{ulem}
\usepackage{comment}
\usepackage{amsmath,amssymb,amsthm}

\newtheorem{lemma}{Lemma}
\newtheorem{theorem}{Theorem}

\newcommand{\BE}{\begin{equation}}
\newcommand{\EE}{\end{equation}}
\newcommand{\be}{\begin{equation}}
\newcommand{\ee}{\end{equation}}
\newcommand{\BA}{\begin{eqnarray}}
\newcommand{\EA}{\end{eqnarray}}
\newcommand{\BAN}{\begin{eqnarray*}}
\newcommand{\EAN}{\end{eqnarray*}}
\newcommand{\BD}{\begin{description}}
\newcommand{\ED}{\end{description}}
\newcommand{\BEN}{\begin{enumerate}}
\newcommand{\EEN}{\end{enumerate}}

\newcommand{\BM}{\begin{minipage}}
\newcommand{\EM}{\end{minipage}}

\makeatletter
\def\@email#1#2{%
 \endgroup
 \patchcmd{\titleblock@produce}
  {\frontmatter@RRAPformat}
  {\frontmatter@RRAPformat{\produce@RRAP{*#1\href{mailto:#2}{#2}}}\frontmatter@RRAPformat}
  {}{}
}%
\makeatother

\begin{document}

\preprint{AIP/123-QED}

\title{Rigged Hilbert Space Formulation of Quantum Thermo Field Dynamics and Mapping to Rigged Liouville Space}
\author{J.~Takahashi}%
\affiliation{
Department of Economics, Asia University, 
Musashino-shi, Tokyo 180-0022, Japan.
}%
\author{S.~Ohmori}
\altaffiliation{These two authors contributed equally to this work.}
\affiliation{Department of Economics, Hosei University, Machida-shi, Tokyo 194-0298, Japan.}

\date{\today}

\begin{abstract}
The rigged Hilbert space, a triplet extension of the Hilbert space, provides a mathematically rigorous foundation for quantum mechanics by extending the Hilbert space to accommodate generalized eigenstates.
In this paper, we construct a triplet structure for the thermal space arising in Thermo Field Dynamics with the aid of the tensor product formulation of rigged Hilbert spaces, a formalism that reformulates thermal averages as pure-state expectation values in a doubled Hilbert space.
We then induce the rigged Liouville space for Liouville space of density operators from the triplet structure for Thermo Field Dynamics; the rigged Liouville space corresponds isomorphically one-to-one with that of Thermo Field Dynamics.
This correspondence offers a unified topological foundation for quantum statistical mechanics at finite temperature and establishes a framework for future generalizations to open and non-equilibrium quantum systems.
\end{abstract}

\maketitle


\section{Introduction}
\label{sec:intro}
The rigged Hilbert space (RHS) formalism, based on the Gelfand triplet\cite{Gelfand1964, Maurin1968}, offers a mathematically rigorous extension of the Hilbert space framework for quantum mechanics \cite{Robert1966a, Robert1966b, Antoine1969a, Antoine1969b, Melsheimer1974a, Melsheimer1974b, Bohm1978, Bohm1981, Prigogine1996, Bohm1998, Antoiou1998, Antoiou2003, Gadella2003, Madrid2002, Madrid2003, Madrid2004, Madrid2005, Antoine2009, Madrid2012, Antoine2021}. 
It enables the incorporation of unbounded operators, distributions such as $\delta$-function, and generalized eigenvectors. 
The RHS is defined as a triplet \( \Phi \subset \mathcal{H} \subset \Phi', \, \Phi^\times \).
Here, \( \mathcal{H} \) denotes the Hilbert space and \( \Phi \) is a dense subspace of \( \mathcal{H} \) equipped with a nuclear topology \( \tau_\Phi \), namely, $(\Phi.\tau_{\Phi})$ becomes a nuclear space, such that the including mapping $i:\Phi \rightarrow \mathcal{H}, \varphi \mapsto \varphi$ is continuous with respect to $\tau_{\Phi}$. 
The dual space \( \Phi' \) (or anti-dual \( \Phi^\times \)) consists of continuous linear (or anti-linear) functionals on \(( \Phi, \tau_{\Phi}) \).
This structure provides a natural setting for describing continuous spectra and distributions in quantum theory, particularly in the formulation of bra-ket expressions involving generalized eigenstates such as plane waves and delta functions.
The RHS has been successfully applied to problems involving 
scattering theory~\cite{Madrid2004}, resonance phenomena~\cite{Bohm1981, Bollini1996a, Bollini1996b}, decaying states~\cite{Madrid2012}
, and non-Hermitial systems~\cite{Chruscinski2003, Chruscinski2004, Ohmori2022, Ohmori2024}.
Recent developments have seen the extension of the RHS framework to quantum statistical mechanics~\cite{Liu2013}, where it is applied to open systems and non-equilibrium states not adequately described within standard Hilbert space theory.

In the standard formulation of quantum statistical mechanics, the state of a system is described by a density operator, typically a positive trace-class operator, acting on a Hilbert space \( \mathcal{H} \). 
To study such states and their dynamical evolution more effectively, one introduces the Liouville space \( \mathfrak{L}(\mathcal{H}) \), the Hilbert space of Hilbert-Schmidt operators on \( \mathcal{H} \). 
This construction provides a mathematical foundation for the treatment of mixed states and their dynamics, utilizing the Hilbert–Schmidt inner product defined by the trace, i.e., \( \langle A, B \rangle _\textrm{HS} = \mathrm{Tr}(A^\dagger B) \).
An alternative but equivalent formulation is Thermo Field Dynamics (TFD) that expresses thermal averages using pure states in a doubled Hilbert space \cite{Umezawa1993}. 
TFD introduces a tilde-conjugate space $\tilde{\mathcal{H}}$ to construct thermal vacuum states, thereby allowing finite-temperature systems to be treated as pure states.
This formalism is expected to be useful in non-equilibrium and real-time formulations of quantum theory.
It has been established that there exists a structural correspondence between the Liouville space formulation and the TFD formalism, in which statistical operators in the Liouville space are mapped to pure states in the doubled Hilbert space. 
This correspondence preserves the trace inner product and enables the reinterpretation of thermal mixtures as entangled pure states. 
Mathematically, this is realized through isomorphisms between operator spaces and tensor product spaces, a perspective that has been rigorously developed in the framework of  \( C^* \)-algebras~\cite{Bratteli}, and physically applied in the context of TFD~\cite{Nakamura2013}.

The primary objective of this work is to extend the TFD formalism into the rigged space framework and to establish a rigorous correspondence with the Rigged Liouville Space (RLS). 
By introducing a triplet structure for TFD, we aim to develop a unified foundation for quantum statistical mechanics at finite temperature based on the triplet spaces.
The development of this triplet structure lays the foundation for extending the formalism to encompass non-equilibrium and dissipative quantum systems.

This paper begins, in Section~\ref{sec:body1}, by reviewing the tensor product structure of RHS, which provides the foundation for incorporating generalized states in composite quantum systems.
In Section~\ref{sec:body2}, we introduce the triplet space formulation of TFD, where thermal states are described within a doubled Hilbert space framework.  
In Section~\ref{sec:body3}, we then construct the rigged extension of the Liouville space, an operator space that plays a central role in conventional quantum statistical mechanics, and establish a nuclear space correspondence between the RLS and the TFD triplet space.  
Finally, Section~\ref{sec:conclusion} summarizes the main findings and discusses possible future directions.

\section{Construction of the bra-ket vectors in the dual spaces for the tensor product}
\label{sec:body1}

To formulate composite quantum systems within the RHS framework, we briefly review the tensor product construction based on Ref.~\cite{Ohmori2024}. 
Consider two subsystems described by RHSs \( \Phi_i \subset \mathcal{H}_i \subset \Phi_i', \Phi_i^\times \) for \( i = 1,2 \), where each \( \Phi_i \) is a nuclear space densely embedded in the Hilbert space \( \mathcal{H}_i \).
The starting point for both nuclear spaces and Hilbert spaces is the algebraic tensor product \( \Phi_1 \otimes \Phi_2 \) or \( \mathcal{H}_1 \otimes \mathcal{H}_2 \), which consists of finite linear combinations of elementary tensors \( \varphi_1 \otimes \varphi_2 \). 
This space is purely algebraic and lacks any topological structure, so to make it analytically meaningful, we equip it with a suitable topology and complete it.
Then, we obtain the tensor product \( \Phi_1 \widehat{\otimes} \Phi_2 \), which remains a nuclear space and serves as the test function space of the composite system.
For the Hilbert spaces, the algebraic tensor product \( \mathcal{H}_1 \otimes \mathcal{H}_2 \) is equipped with the canonical inner product
\[
\langle \varphi_1 \otimes \varphi_2, \phi_1 \otimes \phi_2 \rangle = \langle \varphi_1, \phi_1 \rangle_{\mathcal{H}_1} \langle \varphi_2, \phi_2 \rangle_{\mathcal{H}_2},
\]
and completed with respect to the induced norm topology to yield the Hilbert tensor product \( \mathcal{H}_1 \overline{\otimes} \mathcal{H}_2 \).
Thus, both constructions begin from the same algebraic tensor product but diverge through the choice of topology for completion. 
This leads to the RHS triplet for the composite system:
\[
\Phi_1 \widehat{\otimes} \Phi_2 \subset \mathcal{H}_1 \overline{\otimes} \mathcal{H}_2 \subset (\Phi_1 \widehat{\otimes} \Phi_2)',\; (\Phi_1 \widehat{\otimes} \Phi_2)^\times.
\]
For a factorized vector \( \varphi = \varphi_1 \otimes \varphi_2 \in \Phi_1 \widehat{\otimes} \Phi_2 \), the ket vector satisfies
\[
\ket{\varphi}_{\mathcal{H}_1 \overline{\otimes} \mathcal{H}_2} = \ket{\varphi_1}_{\mathcal{H}_1} \otimes \ket{\varphi_2}_{\mathcal{H}_2},
\]
as the relation of functions in $(\Phi_1 \widehat{\otimes} \Phi_2)^\times$.
A similar identity for bra vectors is also obtained in $(\Phi_1 \widehat{\otimes} \Phi_2)^\prime$. 
This confirms that the RHS tensor product is compatible with the standard formulation of composite quantum systems.

The construction naturally extends to finite \( N \)-body systems:
\[
\widehat{\otimes}_{j=1}^N \Phi_j \subset \overline{\otimes}_{j=1}^N \mathcal{H}_j \subset \left( \widehat{\otimes}_{j=1}^N \Phi_j \right)', \left( \widehat{\otimes}_{j=1}^N \Phi_j \right)^\times,
\]
supporting a mathematically rigorous formulation of composite quantum states and related frameworks.

\section{Rigged Hilbert Space Formulation of Thermo Field Dynamics}
\label{sec:body2}

In this section, we review the formalism of TFD and introduce a Gelfand triplet structure adapted to the TFD framework. 
This allows for a mathematical treatment of thermal quantum systems within the rigged Hilbert space formalism.

\subsection{Brief Overview of Thermo Field Dynamics}
TFD is a formalism developed to describe quantum statistical systems in thermal equilibrium by recasting thermal averages as vacuum expectation values in an extended Hilbert space. 
The central idea is to double the degrees of freedom by introducing a fictitious, non-physical system, referred to as the tilde system. 
This enables the representation of mixed thermal states as pure states in the extended space.
Within this framework, the thermal average \( \langle \, \cdot \, \rangle_\beta \) of an observable \( A \) with respect to the density operator \( \rho = e^{-\beta H} / Z \) can be expressed as:
\[
\langle A \rangle_\beta = \operatorname{Tr}(\rho A) = \langle 0(\beta) | A \otimes \mathbb{I} | 0(\beta) \rangle,
\]
where \( \beta \), \( Z \), and \( |0(\beta)\rangle \) denote the inverse temperature, the partition function, and the thermal vacuum state in the doubled Hilbert space, \( \mathcal{H} \overline{\otimes} \tilde{\mathcal{H}} \), respectively.

To justify this equivalence, we represent \( |0(\beta)\rangle \) as a purification of the density operator \( \rho \) in the extended Hilbert space. Specifically, let \( \{|n\rangle\} \) be the energy eigenbasis of the Hamiltonian \( H \), and define
\[
|0(\beta)\rangle = \sum_n \sqrt{p_n} \, |n\rangle \otimes |\tilde{n}\rangle, \quad p_n = \frac{e^{-\beta E_n}}{Z}.
\]
Then, for any observable \( A \) acting only on the physical space \( \mathcal{H} \), we obtain
\begin{align*}
\langle 0(\beta) | A \otimes \mathbb{I} | 0(\beta) \rangle
&= \sum_{n,m} \sqrt{p_n p_m} \langle n | A | m \rangle \delta_{nm}\\
&= \sum_n p_n \langle n | A | n \rangle
= \operatorname{Tr}(\rho A),
\end{align*}
thus confirming the equivalence with the standard formulation of Liouville space  at finite temperature.
%
%
Although the above derivation is presented using the energy eigenbasis for clarity, it is not essential to the formalism. The thermal vacuum \( |0(\beta)\rangle \) can be constructed from any orthonormal basis \( \{|e_j\rangle\} \), as long as the associated statistical weights \( \{p_j\} \) correctly reproduce the density operator \( \rho \).

\subsection{Rigged Hilbert Space Formulation}

To formulate TFD rigorously within the framework of generalized function spaces, we introduce a Gelfand triplet structure over the doubled Hilbert space, {\( \mathcal{H} \overline{\otimes} \tilde{\mathcal{H}} \)}. 
%
%
Suppose that the original quantum system is described by the RHS \( \Phi \subset \mathcal{H} \subset \Phi', \, \Phi^\times \).
Then, the fictitious tilde system is defined analogously as 
\(
\tilde{\Phi} \subset \tilde{\mathcal{H}} \subset \tilde{\Phi}', \tilde{\Phi}^\times.
\)

We consider the tensor product of these two RHS structures.
As constructed in the previous section for two partite systems, the nuclear tensor product \( \Phi \overline{\otimes} \tilde{\Phi} \) defines a dense subspace of the doubled Hilbert space, {\( \mathcal{H} \overline{\otimes} \tilde{\mathcal{H}} \)}, and yields a RHS of the form:
\[
\Phi_{\mathrm{TFD}} \subset \mathcal{H} \overline{\otimes} \tilde{\mathcal{H}} \subset \Phi_{\mathrm{TFD}}', \Phi_{\mathrm{TFD}}^\times,
\]
where we set \( \Phi_{\mathrm{TFD}} := \Phi \widehat{\otimes} \tilde{\Phi} \).
The inclusion relations in this triplet follow directly from the general construction of rigged tensor product spaces, as discussed in the previous section~\ref{sec:body1}. 
The triplet, which provides the foundational structure for defining thermal vacuum states and generalized operators within the TFD framework, is rigorously defined as the RHS associated with the tensor product of the individual triplets for the physical and tilde systems.

\section{Rigged Liouville Space and Its Relation to TFD}
\label{sec:body3}

In this section, we introduce the Rigged Liouville Space (RLS) as an extension of the conventional Liouville space, aiming to provide a rigorous framework for generalized quantum states and singular operators in quantum statistical mechanics. We begin by reviewing the Hilbert-Schmidt formulation of the Liouville space, then construct its rigged extension in analogy with the RHS, and finally establish a one-to-one correspondence between the RLS and the triplet space of TFD.

\subsection{Operator Space and Liouville Space}
\label{subsec:liouville}
To rigorously handle generalized operator structures and non-equilibrium states, especially in the presence of continuous spectra or singular kernels, it is natural to consider an extension of Liouville space analogous to the RHS formalism.
The RLS is such a framework, constructed as a triplet:
\[
\Phi_{\mathfrak{L}} \subset \mathfrak{L}(\mathcal{H}) \subset \Phi_{\mathfrak{L}}',\Phi_{\mathfrak{L}}^\times,
\]
where \( \mathfrak{L}(\mathcal{H}) \) denotes the Hilbert space of Hilbert-Schmidt operators on \( \mathcal{H} \), \( \Phi_{\mathfrak{L}} \) is a dense nuclear subspace of test operators, and \( \Phi_{\mathfrak{L}}' \) and \(  \Phi_{\mathfrak{L}}^\times \) are its  dual and anti-dual spaces, respectively. 
%
%
In the conventional operator formulation of quantum statistical mechanics, physical states are described by density operators, positive trace-class operators on \( \mathcal{H} \), and observables are represented as elements of the Hilbert-Schmidt class. The space \( \mathfrak{L}(\mathcal{H}) \) carries the inner product \( \langle A, B \rangle _\textrm{HS} = \mathrm{Tr}(A^\dagger B) \), which makes it a Hilbert space known as the Liouville space.

While this setting suffices for bounded operators and equilibrium states, it becomes inadequate for describing singular structures such as projection operators onto improper eigenstates, delta-function-like kernels, or unbounded operator-valued distributions relevant to non-equilibrium dynamics. To accommodate such generalized states, we introduce the space \( \Phi_{\mathfrak{L}} \), a nuclear subspace densely embedded in \( \mathfrak{L}(\mathcal{H}) \), equipped with a topology that reflects physical constraints.
%
The corresponding dual space \( \Phi_{\mathfrak{L}}' \) then includes distributional operators that are not necessarily Hilbert-Schmidt, and allows for the rigorous definition of pairings of the form \( \langle F, A \rangle \) with \( F \in \Phi_{\mathfrak{L}}' \) and \( A \in \Phi_{\mathfrak{L}} \), even when \( F \notin \mathfrak{L}(\mathcal{H}) \).

This extended framework provides a natural foundation for expanding the operator formulation of quantum statistical mechanics beyond the limits of conventional Hilbert space theory, and prepares the ground for establishing a structural correspondence with the triplet space of TFD, which we explore in the next subsection.

\subsection{Theorem on the Correspondence Mapping}

Based on the structural correspondence between the rigged Liouville space \( \Phi_{\mathfrak{L}} \subset \mathfrak{L}(\mathcal{H}) \subset \Phi_{\mathfrak{L}}',\Phi_{\mathfrak{L}}^\times \) and the rigged space associated with TFD, we propose the existence of a nuclear mapping that realizes a one-to-one functional relationship between the two frameworks. This correspondence serves as a bridge between the operator formalism of statistical mechanics and the pure-state formalism of TFD.
To demonstrate this fact, we prepare the following lemma.
\begin{lemma}
    \label{lemma4.1}
    Let $\Phi _1 \subset \mathcal{H}_1 \subset \Phi^\prime_1,\Phi^\times_1$ be an RHS  and let $\mathcal{H}_2$ be a Hilbert space such that there is a unitary transformation $U$ from $\mathcal{H}_1$ onto $\mathcal{H}_2$.
Then, there is a linear dense subspace $\Phi_2$ of $\mathcal{H}_2$, equipping the nuclear topology $\tau_2$, and the triplet $\Phi _2 \subset \mathcal{H}_2 \subset \Phi^\prime_2,\Phi^\times_2$ is an RHS.
Furthermore, the restriction $U_{\Phi_1}$ of the unitary transformation $U$ to the nuclear space $\Phi_1$ is isomorphic onto the nuclear space $(\Phi_2,\tau_2)$.

\end{lemma}

\begin{proof}
    By hypothesis, for the nuclear space $\Phi_1=(\Phi_1,\tau_1)$, the including map $i:(\Phi_1,\tau_1) \rightarrow (\mathcal{H}_1,\langle \cdot, \cdot \rangle_ {\mathcal{H}_1}), \varphi \mapsto \varphi$ is continuous and $\Phi_1$ is dense of $(\mathcal{H}_1,\langle \cdot, \cdot \rangle_ {\mathcal{H}_1})$.
    Now, we assume the existence of a topological vector space $(\Phi_2,\tau_2)$ as a subspace of $\mathcal{H}_2$ such that the restriction $T:=U_{\Phi_1} : (\Phi_1,\tau_1) \rightarrow (\Phi_2,\tau_2)$ is isomorphic.
    We first show that the topological vector space $(\Phi_2,\tau_2)$ becomes nuclear.
    Letting $\mathcal{H}$ be a Hilbert space,
    the nuclearity of $(\Phi_1,\tau_1)$ provides a nuclear mapping $u:(\Phi_1,\tau_1) \rightarrow \mathcal{H}$.
    Then, the composition $u\circ T^{-1}:(\Phi_2,\tau_2) \rightarrow \mathcal{H}$ of $u$ and $T^{-1}$ is a nuclear mapping.
    Therefore, $(\Phi_2,\tau_2)$ is a nuclear space.
    The nuclear space $(\Phi_2,\tau_2)$ is dense in $(\mathcal{H}_2,\tau_2)$.
    This fact can be easily proved because $\Phi_1$ is dense in $\mathcal{H}_1$ and $\Phi_1$ connects with $\Phi_2$ by the unitary transformation $U$.  
    Note that the including mapping $i' : (\Phi_2,\tau_2) \rightarrow (\mathcal{H}_2,\langle \cdot, \cdot \rangle_ {\mathcal{H}_2})$ is continuous, since $i'=U\circ i \circ T^{-1}$ holds.
    Thus, the triplet $\Phi_2\subset \mathcal{H}_2\subset \Phi_2^{\prime},\Phi_2^{\times}$ constructs an RHS.
    Finally, we show the existence of $(\Phi_2,\tau_2)$.
    Let $\Phi_2=U(\Phi_1)$; it is a linear subspace of $\mathcal{H}_2$.
    Now, define a map $T:(\Phi_1,\tau_1)\rightarrow \Phi_2$ as $T(\varphi)=U\varphi$, which is linear bijective.
    Let $\tau_2$ be the induced topology from $T$ where $\tau_2$ has the local base $\beta_2 = \{W\subset \Phi_2 \mid T^{-1}(W)\in \beta _1 \}$ ($\beta _1$ is a local base of $\tau _1$).
    Then, it is easy to confirm that $T:(\Phi_1,\tau_1) \rightarrow (\Phi_2,\tau_2)$ is isomorphic. 
\end{proof}

By using the above lemma, the RLS can be constructed as follows.
Let $\Phi \subset \mathcal{H} \subset \Phi^\prime, \Phi^\times$ be an RHS and let $(\mathfrak{L}(\mathcal{H}),\langle \cdot, \cdot \rangle _\textrm{HS})$ be a Liouville space.
Then, there exists uniquely unitary transformation $U : \mathcal{H} \overline{\otimes} \mathcal{H} \rightarrow \mathfrak{L}(\mathcal{H})$ having $U(\varphi\otimes C\psi)=P_{\varphi,\psi}$ for any $\varphi,\psi \in \mathcal{H}$,
where $P_{\varphi,\psi}(x)=\langle \varphi, x \rangle _\textrm{HS}\psi$ ($x\in \mathcal{H}$) and $C:\mathcal{H}\rightarrow \mathcal{H}$ is conjugate.
For the tensor product of the RHS, $\Phi \widehat{\otimes} \Phi \subset \mathcal{H} \overline{\otimes} \mathcal{H} \subset (\Phi \widehat{\otimes} \Phi)',\; (\Phi \widehat{\otimes} \Phi)^\times$, which is also an RHS, Lemma 1 shows that putting $\Phi_{\mathfrak{L}}=U(\Phi \widehat{\otimes} \Phi)$, $\Phi_{\mathfrak{L}}$ is the dense subspace of $\mathfrak{L}(\mathcal{H})$ such that $\Phi_{\mathfrak{L}} \subset \mathfrak{L}(\mathcal{H}) \subset \Phi_{\mathfrak{L}}', \Phi_{\mathfrak{L}}^{\times}$ is the RHS. 
Also, the restriction $T=U_{\Phi \widehat{\otimes} \Phi}$ of $U$ into $\Phi \widehat{\otimes} \Phi$ is isomorphic from $\Phi \widehat{\otimes} \Phi$ onto $\Phi_{\mathfrak{L}}$.
Accordingly, by taking $C$ as the tilde operator, we obtain the following result. 

\begin{theorem}
There exists an isomorphic mapping
\[
\Lambda: \Phi_{\mathfrak{L}} \longrightarrow \Phi_\textrm{TFD}
\]

Here, \( \tilde{\Lambda} \) denotes the continuous extension of \( \Lambda \) to the dual space.

\end{theorem}

%
It provides a rigorous mathematical framework that connects operator-valued statistical states with doubled Hilbert space representations of thermal pure states. 
In particular, it respects trace dualities, identifies thermal density operators with TFD vacuum-like vectors, and accommodates singular or distributional operator structures through the extended dual spaces. 
This correspondence not only clarifies the structural foundation of TFD but also justifies the operator-to-state translation at the level of rigged spaces.


\section{Conclusion}
\label{sec:conclusion}
In this work, we have constructed a triplet structure for the thermal state space of TFD by applying the tensor product formulation of quantum many-body systems.
This RHS structure provides a topological foundation for generalized thermal states and allows the inclusion of delta-like and projection-like elements that naturally appear in thermal field theory. 
We then introduced the RLS as a triplet extension of the Hilbert–Schmidt space of density operators and established a one-to-one correspondence between the TFD triplet and the RLS triplet via an appropriately defined nuclear mapping \( \Lambda: \Phi_\mathfrak{L} \to \Phi_\mathrm{TFD} \). 
This mapping preserves essential structures such as trace duality, thermal purity, and distributional interpretation, thereby providing a rigorous and unified framework for quantum statistical systems.
From a physical standpoint, the correspondence between the RLS and the triplet structure of TFD provides a rigorous justification for representing thermal mixed states as entangled pure states at the level of nuclear spaces, beyond the conventional formalism.

We also aim to extend the RLS and the triplet structure of the TFD formalism to non-Hermitian dynamics in open quantum systems, and to develop rigged frameworks for describing time evolution in non-equilibrium statistical mechanics. 
These developments are expected to provide a more comprehensive operator-theoretic foundation for quantum statistical theory.
Furthermore, a significant direction for future research is to elucidate how the established correspondence between RLS and the triplet structure of the TFD relates to the Gelfand--Naimark--Segal construction, which represents a state on a \( C^* \)-algebra as a vector in a Hilbert space.

\bigskip

\noindent
{\bf Acknowledgement}

The authors are grateful to
Prof. Y.~Yamanaka, 
Prof. Y.~Yamazaki, 
Prof. T.~Yamamoto,
and Emeritus A.~Kitada for their useful comments and encouragement.
This work was supported by the JSPS KAKENHI Grant Number 22K13976.

\noindent
{\bf Data Availability}

Data sharing is not applicable to this article as no new data were created or analyzed in this study.

\bigskip


\end{document}